\documentclass[letterpaper, 10 pt, conference]{ieeeconf}

\IEEEoverridecommandlockouts
\usepackage{graphics}
\usepackage{epsfig}
\usepackage{mathptmx}
\usepackage{times}
\usepackage{amsmath}
\usepackage{amssymb}
\usepackage{color}
\usepackage{subfig}
\usepackage{mathrsfs}
\usepackage{booktabs}
\usepackage{epstopdf}
\usepackage{multirow}
\usepackage{fancyhdr}

\newtheorem{prop}{Proposition}[section]

\newtheorem{remark}{Remark}[section]

\title{\LARGE \bf Cooperative Control of TCSC to Relieve the Stress of Cyber-physical Power System}

\author{Chao~Zhai, Gaoxi~Xiao, Hehong Zhang and Tso-Chien Pan \thanks{Chao Zhai, Gaoxi Xiao, Hehong Zhang and Tso-Chien Pan are with Institute of Catastrophe Risk Management, Nanyang Technological University, 50 Nanyang Avenue, Singapore 639798. They are also with Future Resilient Systems, Singapore-ETH Centre, 1 Create Way, CREATE Tower, Singapore 138602. Chao Zhai, Gaoxi Xiao and Hehong Zhang are also with School of Electrical and Electronic Engineering, Nanyang Technological University. Corresponding author: Gaoxi Xiao. Email: {\tt\small egxxiao@ntu.edu.sg}}
}

\begin{document}

\maketitle
\thispagestyle{empty}
\pagestyle{empty}

\begin{abstract}
This paper addresses the cooperative control problem of Thyristor-Controlled Series Compensation (TCSC) to eliminate the stress of cyber-physical power system. The cyber-physical power system is composed of power network, protection and control center and communication network. A cooperative control algorithm of TCSC is developed to adjust the branch impedance and regulate the power flow. To reduce computation burdens, an approximate method is adopted to estimate the Jacobian matrix for the generation of control signals. In addition, a performance index is introduced to quantify the stress level of power system. Theoretical analysis is conducted to guarantee the convergence of performance index when the proposed cooperative control algorithm is implemented. Finally, numerical simulations are carried out to validate the cooperative control approach on IEEE 24 Bus Systems in uncertain environments.
\end{abstract}

\section{INTRODUCTION}\label{sec:int}
Due to huge economic loss, power system blackouts have become an issue of great concern to both electrical power industries and governments in the past several decades. Thus, a lot of efforts are taken to develop various protection schemes against the blackout. Some researchers are dedicated to the prevention of cascading failures when the cascade propagates in the early stage \cite{beg05,cz17}, while others focus on the identification of initial contingencies \cite{kim16,czm17}. Actually, most blackouts are closely related to the precondition of excessive power demand, which results in the stress of power system. For instance, more than 60\% blackouts take place in the summer and winter peaks when the power demand is relatively high \cite{lu06}. Therefore, it is significant to eliminate the undesired consequence of preconditions (e.g., power system stress due to branch overloads) and reduce the risk of cascading blackouts.

The application of flexible alternate current transmission systems (FACTS) greatly improves the performance of power system in terms of power oscillation damping and transient stability enhancement. As an important member in the FACTS family, Thyristor-Controlled Series Compensation (TCSC) plays a major part in the reliable operation of power transmission system to relieve the power system stress, provide the voltage support, schedule the power flow, etc \cite{xiao99}. As a result, various control schemes of TCSC have been proposed in the past decades, and they include PID control \cite{pan11}, Fuzzy logic control \cite{hiy95}, energy function method \cite{gro95}, auto-disturbance rejection control (ADRC) \cite{zhang99}, neural network based control \cite{dai98}, to name just a few. Nevertheless, most control schemes focus on the separate operation of TCSC controllers to improve the transient stability and dynamic performance of power system, which ignores the collaboration of TCSC controllers to achieve the better performance.

On the other side, cooperative control of multi-agent system has attracted much interest of researchers in the field of control and systems engineering \cite{cz13,cz16}. Essentially, cooperative control refers to control actions that aim to achieve the control goal through sharing the information of multiple components in a cooperative way. Actually, cooperative control is also applied to the power system protection by regarding each TCSC as an agent that is able to collaborate with each other for scheduling the power flow. Specifically, cooperative control allows the TCSC agents to work together for a common goal by reconciling the conflict of interest among individual TCSCs, which helps to achieve the desired performance in a shorter time. Moreover, it contributes to strengthening the capability of power system against malicious disturbances by absorbing the stress or damage in a systematic manner. Therefore, a cooperative control scheme is proposed in this paper to deal with the problem of power system stress. The main contributions of this work are listed as follows:
\begin{enumerate}
  \item Develop a cooperative control algorithm of TCSC with the guaranteed convergence of performance index in theory.
  \item Propose a simple and efficient approach to estimating the Jacobian matrix, which greatly reduces the computation burdens.
  \item Validate the cooperative control approach on the standard IEEE bus system with unknown load variations.
\end{enumerate}
The outline of this paper is organized as follows. Section \ref{sec:pre} presents the problem formulation of eliminating power system stresses with TCSC. Section \ref{sec:ana} provides the cooperative control algorithm and theoretical analysis. Simulations and validation on IEEE $24$ Bus System are conducted in Section \ref{sec:sim}. Finally, we conclude the paper and discuss future work in Section \ref{sec:con}.

\section{PROBLEM FORMULATION}\label{sec:pre}
\begin{figure}
\scalebox{0.28}[0.28]{\includegraphics{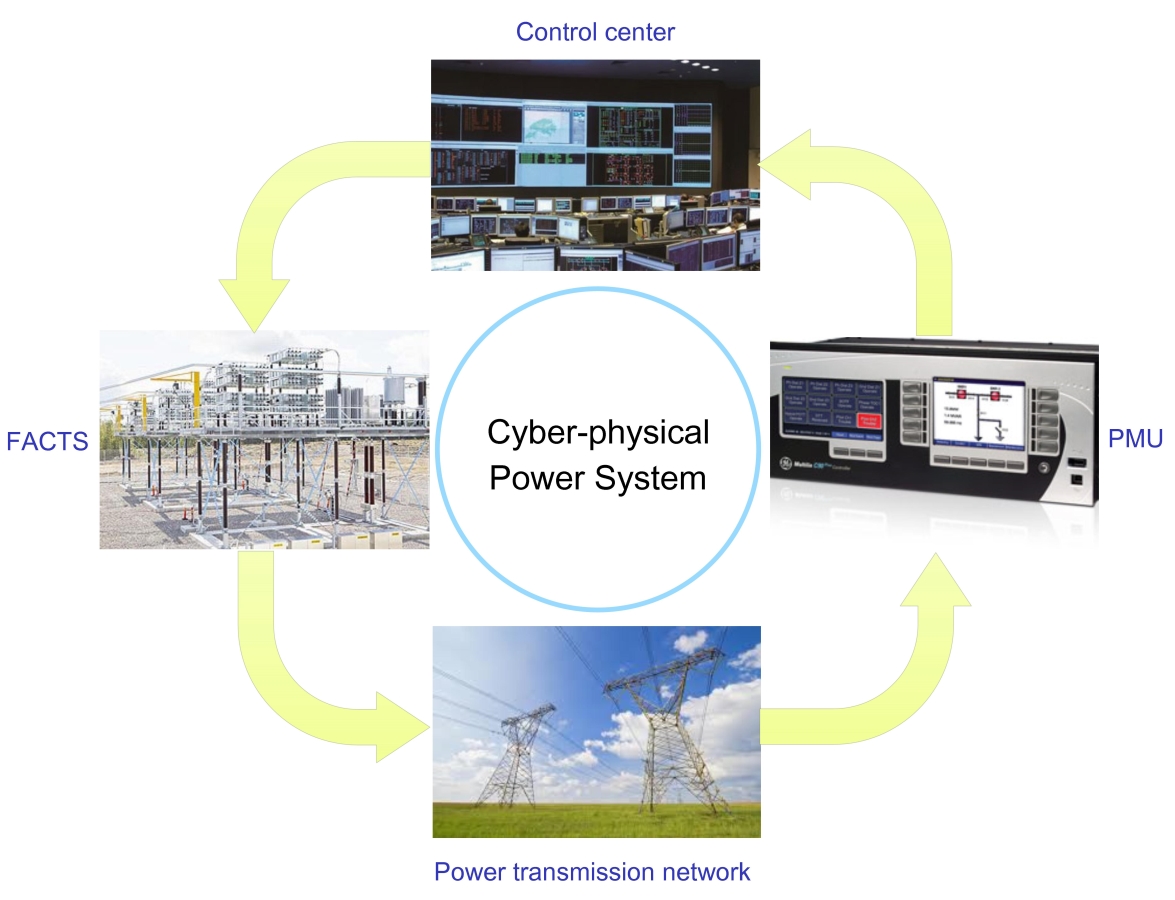}}\centering
\caption{\label{cps} Cyber-physical power system.}
\end{figure}
Figure \ref{cps} presents a schematic diagram of cyber-physical power system, which includes power transmission network, FACTS devices, phasor measurement unit (PUM) and control center. Specifically, the PMUs monitor the state of power system and transmit the state information to control center through communication networks. The control center capitalizes on the state information to generate the control signal for various actuators in power system. As a type of actuators, FACTS devices (e.g., TCSC) adjust the branch impedance according to the control signal.

Consider a power network with $m$ buses and $n$ branches. Let $Z\in C^n$ denote the vector of branch impedance in power system. $P_b\in C^m$ and $P_e=(P_{ij})\in C^n$ represent the vector of injected power on buses and the vector of transmission power on branches, respectively. For Bus $i\in\{1,2,...,m\}$, the AC power flow equation is given as follows.
\begin{equation}\label{pf_p}
Re(P_{b,i})=\sum_{j=1}^{m}|V_i||V_j|\left(G_{ij}\cos\theta_{ij}+B_{ij}\sin\theta_{ij}\right)
\end{equation}
and
\begin{equation}\label{pf_q}
Im(P_{b,i})=\sum_{j=1}^{m}|V_i||V_j|\left(G_{ij}\sin\theta_{ij}-B_{ij}\cos\theta_{ij}\right)
\end{equation}
where $V_i$ and $V_j$ denote the voltage of Bus $i$ and Bus $j$, respectively. $\theta _{ij}$ is the difference of voltage angle between Bus $i$ and Bus $j$. In addition, $G_{ij}$ and $B_{ij}$ refer to the conductance and susceptance of branch connecting Bus $i$ to Bus $j$, respectively. Then the transmission power on the above branch can be computed by
\begin{equation}\label{pf}
P_{ij}=|V_i-V_j|^2Y^*_{ij}
\end{equation}
where $Y^*_{ij}$ represents the the complex conjugate of branch admittance $Y_{ij}$. Suppose that there is the desired transmission power on each branch, denoted by the vector $\sigma=(\sigma_1,\sigma_2,...,\sigma_n)\in C^n$. When the actual transmission power deviates from the desired transmission power, the TSCS agents start to update the branch impedance in order to drive the actual transmission power towards the desired transmission power on each branch.
\begin{figure}
\scalebox{0.055}[0.055]{\includegraphics{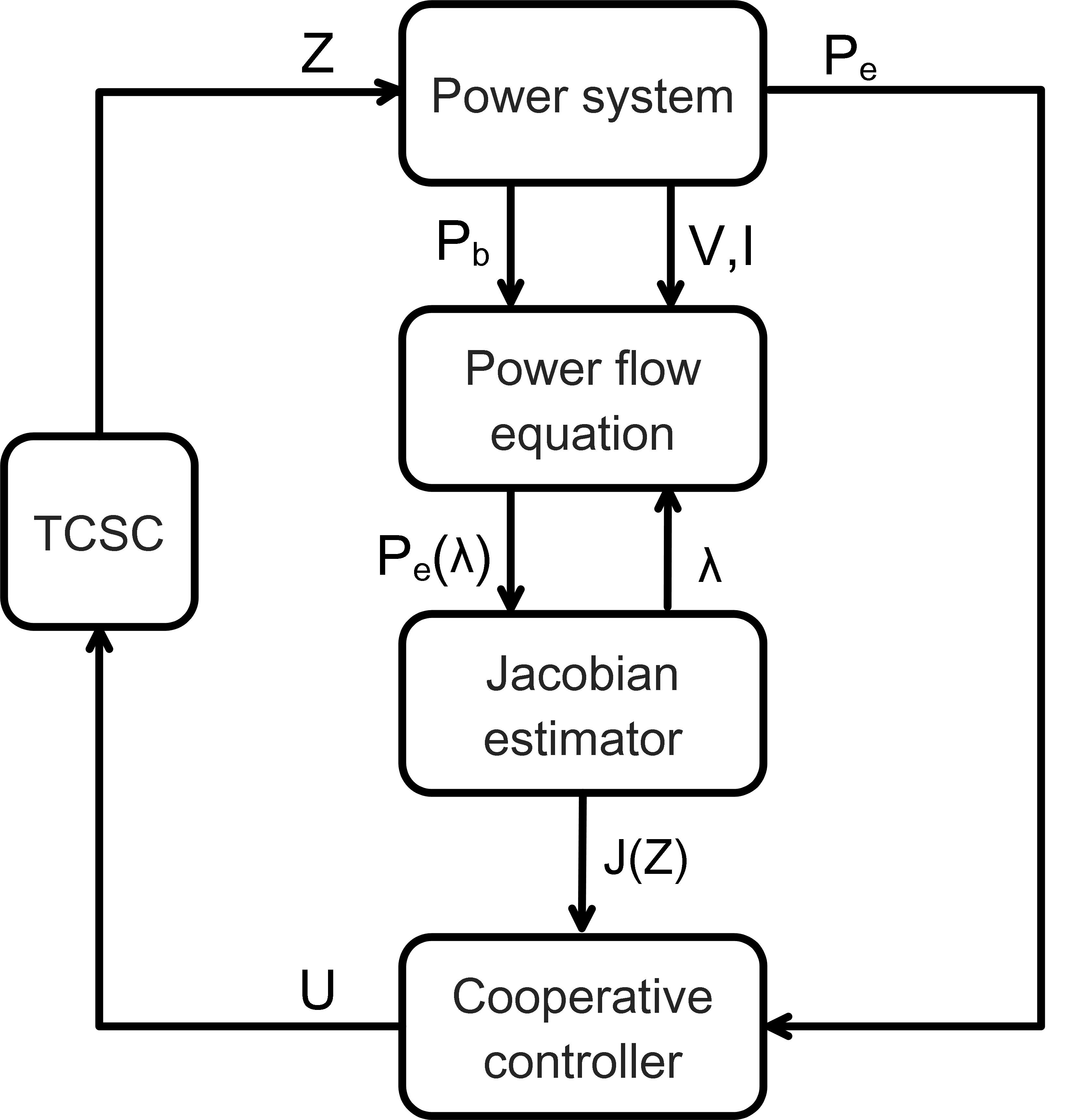}}\centering
\caption{\label{flow} Block diagram of information flow in cyber-physical power systems.}
\end{figure}

Essentially, the goal of this study is to design the control signal $U$ for TCSC in order to minimize the deviation of the actual transmission power from the desired transmission power.
In this way, the stress of power system can be eliminated. Thus, the optimization formulation is presented as follows.
\begin{equation}\label{min}
    \min_{U} H(Z)
\end{equation}
where the objective function is given by
\begin{equation}\label{obj}
    H(Z)=\|Re(P_e)-Re(\sigma)\|^2+\epsilon\|Im(P_e)-Im(\sigma)\|^2
\end{equation}
with the constant $\epsilon\in[0,1]$. The first term in (\ref{obj}) accounts for the mismatch between the actual active power and the desired active power on branches, and the second term describes the error of reactive powers. Additionally, the weight $\epsilon$ is used to quantify the significance of reactive power compared to the active power.

The time derivative of H(Z) along the dynamics of $Re(Z)$ and $Im(Z)$ is given by
\begin{equation}\label{dHY}
\begin{split}
\frac{dH(Z)}{dt}&=2\left(Re(P_e)-Re(\sigma)\right)^T\frac{dRe(P_e)}{dt}\\
                &+2\epsilon\left(Im(P_e)-Im(\sigma)\right)^T \frac{dIm(P_e)}{dt}
\end{split}
\end{equation}
with
\begin{equation}\label{dRe}
\frac{dRe(P_e)}{dt}=\frac{\partial Re(P_e)}{\partial Re(Z)}\cdot\frac{dRe(Z)}{dt}+\frac{\partial Re(P_e)}{\partial Im(Z)}\cdot\frac{dIm(Z)}{dt}
\end{equation}
and
\begin{equation}\label{dIm}
\frac{dIm(P_e)}{dt}=\frac{\partial Im(P_e)}{\partial Re(Z)}\cdot\frac{dRe(Z)}{dt}+\frac{\partial Im(P_e)}{\partial Im(Z)}\cdot\frac{dIm(Z)}{dt}
\end{equation}
By substituting (\ref{dRe}) and (\ref{dIm}) into (\ref{dHY}), we obtain
\begin{equation*}
     \begin{split}
         \frac{dH(Z)}{dt}&=2\left(Re(P_e)-Re(\sigma)\right)^T\left(\frac{\partial Re(P_e)}{\partial Re(Z)}\cdot\frac{dRe(Z)}{dt}\right) \\
          &+2\left(Re(P_e)-Re(\sigma)\right)^T\left(\frac{\partial Re(P_e)}{\partial Im(Z)}\cdot\frac{dIm(Z)}{dt}\right) \\
          &+2\epsilon\left(Im(P_e)-Im(\sigma)\right)^T\left(\frac{\partial Im(P_e)}{\partial Re(Z)}\cdot\frac{dRe(Z)}{dt}\right) \\
          &+2\epsilon\left(Im(P_e)-Im(\sigma)\right)^T\left(\frac{\partial Im(P_e)}{\partial Im(Z)}\cdot\frac{dIm(Z)}{dt}\right)\\
\end{split}
\end{equation*}
which is equivalent to
\begin{equation}\label{dHY1}
\begin{split}
\frac{dH(Z)}{dt}&=2\left(Re(P_e)-Re(\sigma)\right)^T\frac{\partial Re(P_e)}{\partial Re(Z)}\cdot\frac{dRe(Z)}{dt} \\
     &+2\epsilon\left(Im(P_e)-Im(\sigma)\right)^T\frac{\partial Im(P_e)}{\partial Re(Z)}\cdot\frac{dRe(Z)}{dt} \\
          &+2\left(Re(P_e)-Re(\sigma)\right)^T\frac{\partial Re(P_e)}{\partial Im(Z)}\cdot\frac{dIm(Z)}{dt} \\
          &+2\epsilon\left(Im(P_e)-Im(\sigma)\right)^T\frac{\partial Im(P_e)}{\partial Im(Z)}\cdot\frac{dIm(Z)}{dt} \\
          &=2\left[
               \begin{array}{c}
                Re(P_e)-Re(\sigma) \\
                \epsilon(Im(P_e)-Im(\sigma)) \\
               \end{array}
             \right]^TJ(Z)\left[
                               \begin{array}{c}
                                 \frac{dRe(Z)}{dt} \\
                                 \frac{dIm(Z)}{dt} \\
                               \end{array}
                             \right]
     \end{split}
\end{equation}
where the Jacobian matrix $J(Z)$ in (\ref{dHY1}) is denoted by
$$
J(Z)=\left[
                        \begin{array}{cc}
                          \frac{\partial Re(P_e)}{\partial Re(Z)} &  \frac{\partial Re(P_e)}{\partial Im(Z)} \\
                          \frac{\partial Im(P_e)}{\partial Re(Z)} &  \frac{\partial Im(P_e)}{\partial Im(Z)} \\
                        \end{array}
                      \right]=(J_{i,j}(Z))\in R^{2n\times2n}
$$

Let $U_{re}\in R^n$ and $U_{im}\in R^n$ denote the control signals to update the branch resistance and branch reactance, respectively. And the cooperative control input $U$
is composed of $U_{re}$ and $U_{im}$ as follows
\begin{equation}\label{odeu}
U=\left(\begin{array}{c}
    U_{re} \\
    U_{im}
  \end{array}
  \right)
\end{equation}
with
\begin{equation}\label{ode}
    \frac{dRe(Z)}{dt}=U_{re},\quad \frac{dIm(Z)}{dt}=U_{im}
\end{equation}

Figure \ref{flow} presents the information flow on the cooperative control of TCSC in cyber-physical power systems.
Specifically, PMUs detect the injected power $P_b$, the voltage $V$ and the injected current $I$ on buses and send these data to control centre, where the branch impedance can be identified and the power flow is computed according to Equations (\ref{pf_p}), (\ref{pf_q}) and (\ref{pf}). Then the Jacobian matrix is estimated by adding a tiny perturbation $\lambda$ on each branch in turns. Next, the cooperative controller integrates the Jacobian matrix with the error between the actual transmission power and the desired one to produce the control signal of TCSC. Finally, TSCSs adjust the branch impedance based on control inputs $U$ to schedule the power flow $P_e$.

\section{COOPERATIVE CONTROL ALGORITHM}\label{sec:ana}
This section presents the cooperative control algorithm of TSCS and the method of estimating the Jacobian matrix $J(Z)$, as mentioned in Section \ref{sec:pre}.
\subsection{Control law}
The cooperative controller of TCSC is designed as follows
\begin{equation}\label{control}
U=\left(
  \begin{array}{c}
    U_{re} \\
    U_{im} \\
  \end{array}
\right)=-\kappa(Z)\circ J(Z)^T \left[
               \begin{array}{c}
                Re(P_e)-Re(\sigma) \\
                \epsilon(Im(P_e)-Im(\sigma)) \\
               \end{array}
             \right],
\end{equation}
where the symbol $\circ$ denotes the Hadamard product, and $\kappa(Z)=(\kappa_1(Z),\kappa_2(Z),...,\kappa_{2n}(Z))^T$ is the nonnegative vector function.  Each element of $\kappa(Z)$ is given by
$$
\kappa_i(Z)=\left\{
              \begin{array}{ll}
                c, & \hbox{$Z_i\in[\underline{Z}_i,\bar{Z}_i]$;} \\
                0, & \hbox{$otherwise$.}
              \end{array}
            \right.
$$
with the positive constant $c$ and the lower bound $\underline{Z}_i$ and the upper bound $\bar{Z}_i$ for the impedance of Branch $i$. In theory, the proposed control algorithm can guarantee the asymptotical convergence of (\ref{obj}).
\begin{prop}
The control law (\ref{control}) ensures the convergence of Objective Function (\ref{obj}).
\end{prop}
\begin{proof}
Since $Re(Z)$ and $Im(Z)$ are updated according to the control law (\ref{control}), Equation (\ref{dHY1}) can be rewritten as
\begin{equation}\label{dHU}
    \frac{dH(Z)}{dt}=2\left[
               \begin{array}{c}
                Re(P_e)-Re(\sigma) \\
                \epsilon(Im(P_e)-Im(\sigma)) \\
               \end{array}
             \right]^T\left[
                        \begin{array}{cc}
                          \frac{\partial Re(P_e)}{\partial Re(Z)} &  \frac{\partial Re(P_e)}{\partial Im(Z)} \\
                          \frac{\partial Im(P_e)}{\partial Re(Z)} &  \frac{\partial Im(P_e)}{\partial Im(Z)} \\
                        \end{array}
                      \right]U
\end{equation}
By substituting (\ref{control}) into (\ref{dHU}), we obtain
\begin{equation*}
\begin{split}
\frac{dH(Z)}{dt}&=-2\left[
               \begin{array}{c}
                Re(P_e)-Re(\sigma) \\
                \epsilon(Im(P_e)-Im(\sigma)) \\
               \end{array}
             \right]^TJ(Z)\cdot\kappa(Z)~\circ \\
         &J(Z)^T\left[
               \begin{array}{c}
                Re(P_e)-Re(\sigma) \\
                \epsilon(Im(P_e)-Im(\sigma)) \\
               \end{array}
             \right]\\
&=-2\left\|\bar{\kappa}(Z)\circ J(Z)^T \left[
               \begin{array}{c}
                Re(P_e)-Re(\sigma) \\
                \epsilon(Im(P_e)-Im(\sigma)) \\
               \end{array}
             \right]\right\|^2\leq 0
\end{split}
\end{equation*}
with $\bar{\kappa}(Z)=(\sqrt{\kappa_1(Z)},\sqrt{\kappa_2(Z)},...,\sqrt{\kappa_{2n}(Z)})^T\in R^{2n}$. Thus, $H(Z)$ decreases monotonously as time goes to infinity. Moreover, it follows from $H(Z)\geq0$ that the convergence of $H(Z)$ is guaranteed. The proof is thus completed.
\end{proof}
\begin{figure}
\scalebox{0.038}[0.038]{\includegraphics{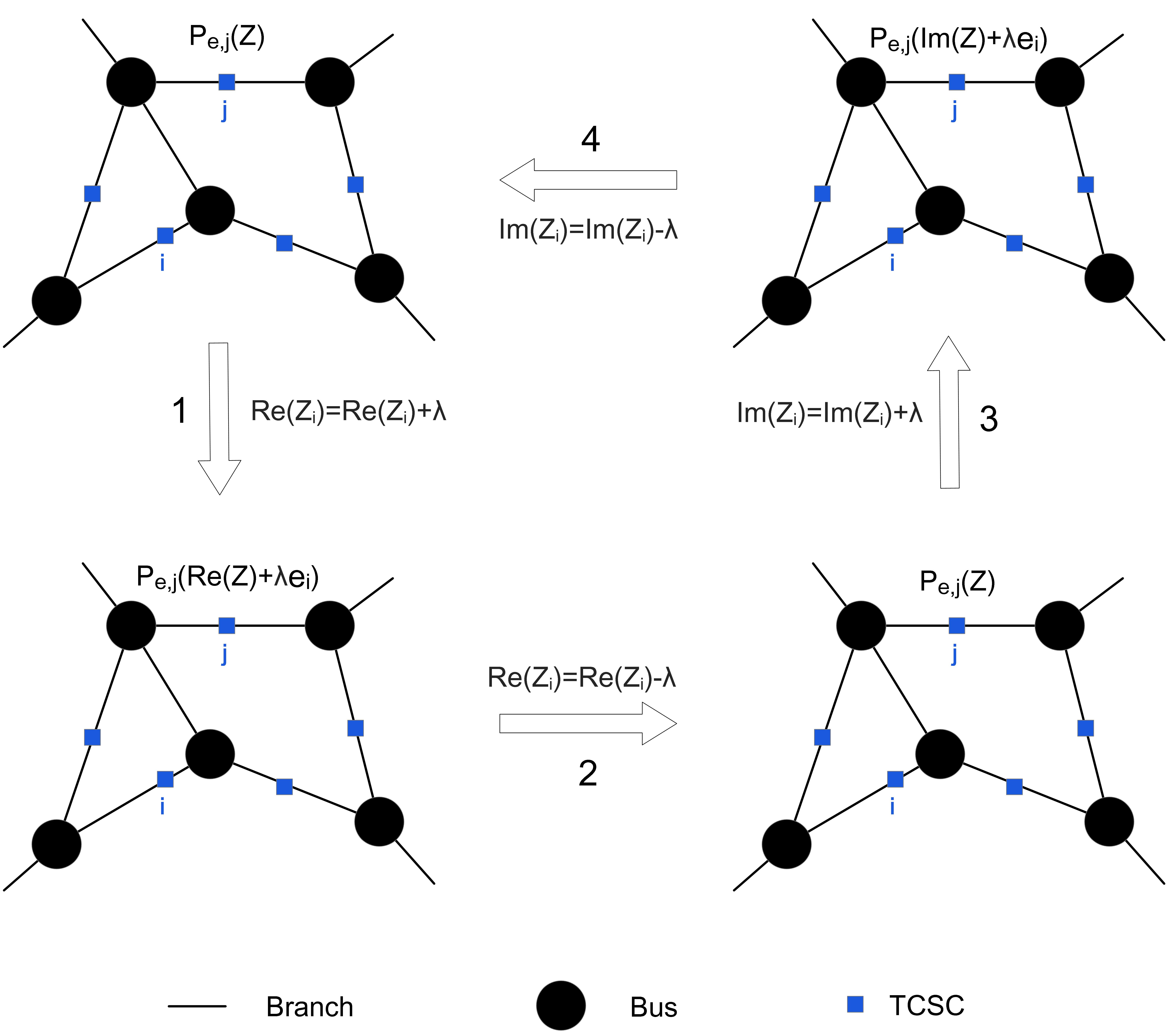}}\centering
\caption{\label{est} Estimation of the Jacobian matrix.}
\end{figure}

\subsection{Jacobian estimator}
In practice, it is difficult to implement the cooperative control algorithm (\ref{control}) in real time, because the accurate Jacobian matrix $J(Z)$ is not directly available. Thus, it is necessary to propose a numerical method to estimate the Jacobian matrix $J(Z)$ with low computation costs. Now, we introduce the approach to approximate the Jacobian matrix $J(Z)$, which includes $4$ steps (see Fig.~\ref{est}). First of all, we compute the transmission power $P_{e}(Z)$ on each branch with the power flow equation. Then Branch $i\in \{1,2,...,n\}$ is selected to increase its resistance by a sufficiently small value $\lambda$, and the resulted transmission power $P_{e}(Re(Z)+\lambda e_i)$ is obtained. Here $e_i$ denotes the $n$-dimensional unit vector with the $i$-th element being $1$ and all other elements being $0$. Thus, the elements in the $i$-th column of Jacobian matrix can be estimated as follows
\begin{equation}\label{ji}
J_{j,i}(Z)\approx\frac{Re(P_{e,j}(Re(Z)+\lambda e_i))-Re(P_{e,j}(Z))}{\lambda},
\end{equation}
and
\begin{equation}\label{jni}
J_{j+n,i}(Z)\approx\frac{Im(P_{e,j}(Re(Z)+\lambda e_i))-Im(P_{e,j}(Z))}{\lambda}.
\end{equation}
where $j\in \{1,2...,n\}$. Next, the resistance of Branch $i$ is restored to the original value. Afterwards, the reactance of the $i$-th branch is increased by $\lambda$, and the resulting transmission power $P_{e}(Im(Z)+\lambda e_i)$ is obtained.
The elements in the ($i+n$)-th column of the Jacobian matrix can be estimated as follows
\begin{equation}\label{jin}
J_{j,i+n}(Z)\approx\frac{Re(P_{e,j}(Im(Z)+\lambda e_i))-Re(P_{e,j}(Z))}{\lambda}
\end{equation}
and
\begin{equation}\label{jnin}
J_{j+n,i+n}(Z)\approx\frac{Im(P_{e,j}(Im(Z)+\lambda e_i))-Im(P_{e,j}(Z))}{\lambda}.
\end{equation}
where $j\in \{1,2...,n\}$. The approximated Jacobian matrix $J(Z)$ is available after implementing the above series of operations for each branch. The procedure of estimating the Jacobian matrix $J(Z)$ is summarized in Table~\ref{ejm}.
\begin{table}
\caption{\label{ejm} Jacobian Estimation Algorithm.}
\begin{center}
\begin{tabular}{lcl} \hline
\bf{Input}: $\lambda$ and $P_b$ ~~~~\bf{Output}: $J(Z)$ \\\hline
  1:~~Set $\lambda$ and detect $P_b$  \\
  2:~~Compute $P_e(Z)$ with (\ref{pf_p}), (\ref{pf_q}) and (\ref{pf}) \\
  3:~\bf{for}~$i=1$ to $n$ do \\
  4:~~~~~~$\texttt{Re}(Z)=\texttt{Re}(Z)+\lambda e_i$ \\
  5:~~~~~~Compute $P_e(Z)$ with (\ref{pf_p}), (\ref{pf_q}) and (\ref{pf}) \\
  6:~~~~~~Estimate $J_{j,i}(Z)$ with (\ref{ji}) and (\ref{jni}) \\
  7:~~~~~~$\texttt{Re}(Z)=\texttt{Re}(Z)-\lambda e_i$ \\
  8:~~~~~~$\texttt{Im}(Z)=\texttt{Im}(Z)+\lambda e_i$ \\
  9:~~~~~~Compute $P_e(Z)$ with (\ref{pf_p}), (\ref{pf_q}) and (\ref{pf}) \\
  10:~~~~~Estimate $J_{j,i+n}(Z)$ with (\ref{jin}) and (\ref{jnin}) \\
  11:~~~~~$\texttt{Im}(Z)=\texttt{Im}(Z)-\lambda e_i$ \\
  12:~\bf{end for} \\ \hline
\end{tabular}
\end{center}
\end{table}

\subsection{Implementation}
In order to further reduce the computation cost, it is unnecessary to update the estimated Jacobian matrix $J(Z)$ at each time step in practice. For this reason, we introduce the performance index:
\begin{equation}\label{index}
    S_k=\max_{i\in I(k)} H_i(Z),
\end{equation}
where $H_i(Z)$ denotes the value of Objective Function (\ref{obj}) at the $i$-th time step. Actually, $S_k$ represents the maximum value of Objective Function (\ref{obj}) in the interval $I(k)=[k-1)T+1,kT]$, $k\in\mathbb{Z}^+$, and $T$ refers to the number of time steps in the interval. When the condition $S_{k+1}\geq S_k$ holds, the proposed control algorithm fails to effectively decrease the mismatch between the actual transmission power and the desired one. At this time, the Jacobian matrix should be updated to adjust the evolution direction.

\begin{table}
\caption{\label{cca} Cooperative Control Algorithm.}
\begin{center}
\begin{tabular}{lcl} \hline
\bf{Input}: $s=0$, $k=0$ \texttt{and} $S_{0}=H_0(Z)$ \\
\bf{Output}: $Z$ \texttt{and} $P_e$ \\\hline
  1:~\bf{while}~($H_s(Z)\neq0$) \\
  2:~~~~~~Detect $P_{e}$ \\
  3:~~~~~~\bf{if}~($\texttt{mod}(s,T)=0$) \\
  4:~~~~~~~~~~Update $k=k+1$ \\
  5:~~~~~~~~~~Compute $S_{k}$ with (\ref{index}) \\
  6:~~~~~~~~~~\bf{if}~($S_{k}\geq S_{k-1}$) or ($s=0$) \\
  7:~~~~~~~~~~~~~~$S_{k}\leftarrow S_{k-1}$ \\
  8:~~~~~~~~~~~~~~Detect $P_b$, $V$ and $I$ \\
  9:~~~~~~~~~~~~~~Run the JEA for $J(Z)$ \\
  10:~~~~~~~~~~\bf{end if} \\
  11:~~~~~~\bf{end if} \\
  12:~~~~~Update $Z$ with (\ref{ode}) and (\ref{control}) \\
  13:~~~~~Update $s=s+1$ \\
  14:~~~~~Compute $H_s(Z)$ with (\ref{obj}) \\
  15:~\bf{end while} \\ \hline
\end{tabular}
\end{center}
\end{table}

Table \ref{cca} describes the procedure of implementing the cooperative control algorithm of TCSC in practice. First of all, PMUs detect the transmission power $P_e$ on branches and send the data to the control center. For every $T$ time steps (i.e., $\texttt{mod}(s,T)=0$), Performance Index $S_k$ is computed with (\ref{index}). If $S_k$ does not decrease compared with the previous value $S_{k-1}$, the Jacobian matrix $J(Z)$ is updated with the Jacobian Estimation Algorithm (JEA) in Table~\ref{ejm}. Afterwards, TCSCs adjust the branch impedance according to the cooperative control law (\ref{ode}) and (\ref{control}). Finally, the above iteration process is repeated until the objective function reaches zero. The cooperative control algorithm in Table~\ref{cca} allows us to gradually decrease the performance index (\ref{index}) by selecting the appropriate Jacobian matrix $J(Z)$.
\begin{prop}\label{pro_cca}
Cooperative control algorithm of TCSC in Table \ref{cca} ensures the monotonous convergence of Performance Index (\ref{index}).
\end{prop}

\begin{proof}
Cooperative control algorithm in Table \ref{cca} allows us to obtain a sequence $\{S_k\}_{k=1}^{\infty}$ with the constraint $S_{k+1}\leq S_k$. It follows from $S_k\geq0$, $\forall k\in \mathbb{Z}^+$ that $S_k$ monotonously converges to $S^*=\inf_{k\in \mathbb{Z}^+}S_k$ as $k\rightarrow+\infty$. The proof is thus completed.
\end{proof}

\begin{remark}
When the sequence $\{S_k\}_{k=1}^{N}$ is available, it is feasible to obtain the upper bound of Objective Function (\ref{obj}) as follows:
$$
0\leq H_i(Z)\leq \max_{i\in I(k)} H_i(Z)=S_k\leq S_N, \quad i\in I(k), \quad k\geq N.
$$
Moreover, $H_i(Z)$ converges to $0$ as $S_k\rightarrow0$, $k\rightarrow+\infty$.
\end{remark}

\section{NUMERICAL SIMULATIONS}\label{sec:sim}
\begin{figure}
\scalebox{0.052}[0.052]{\includegraphics{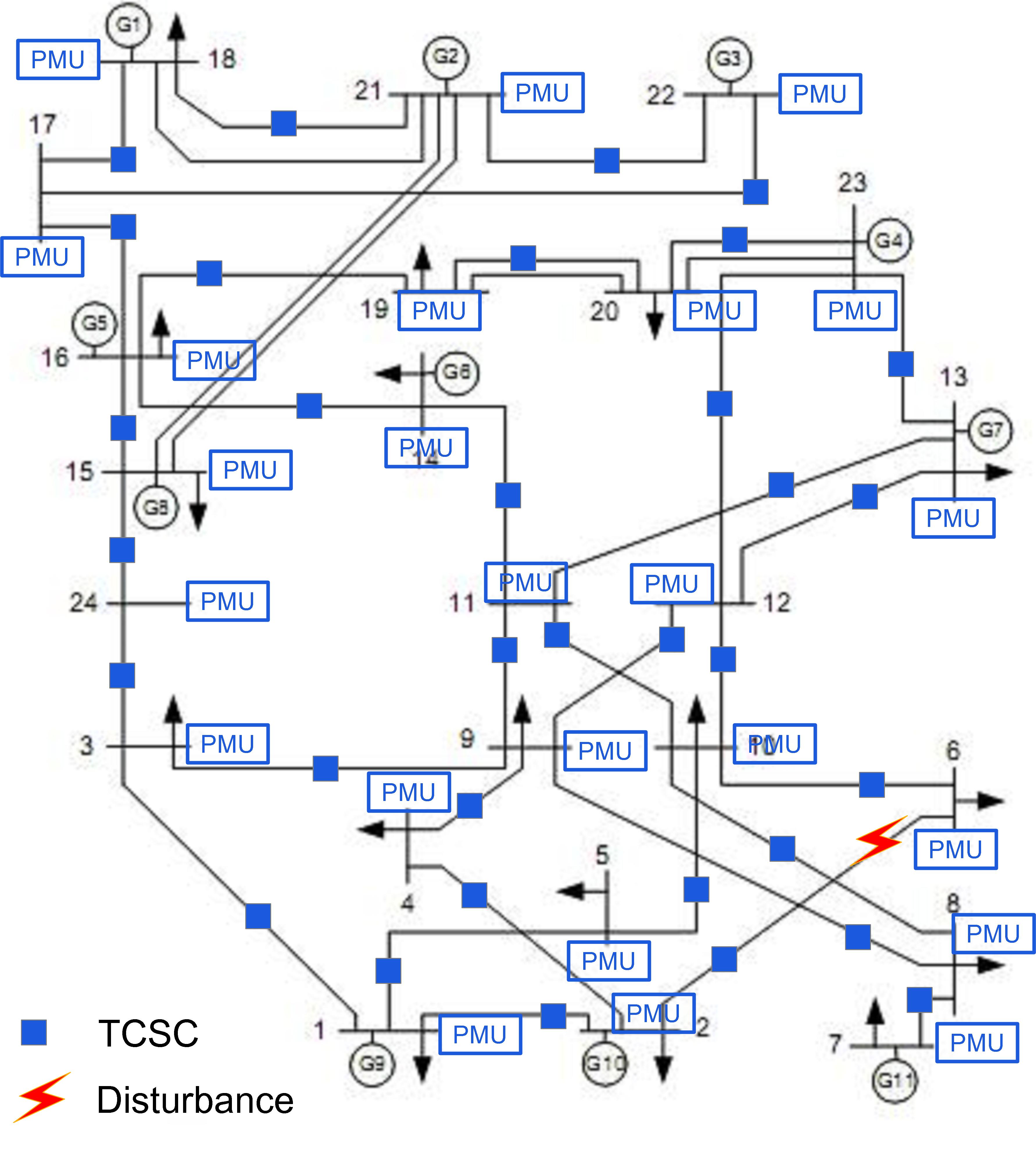}}\centering
\caption{\label{bus24} IEEE 24 Bus System with FACTS devices and PMUs.}
\end{figure}
In this section, numerical simulations are conducted to validate the proposed cooperative control algorithm on IEEE 24 Bus System with time-varying loads (see Fig.~\ref{bus24}). The parameters are selected as follows: the control gain $c=0.02$, the coefficient for reactive power mismatch in objective function $\epsilon=0.2$, the disturbance magnitude $\lambda=10^{-6}$ and $T=100$. For simplicity, Euler method is adopted to solve the differential equation (\ref{ode}) with the step size $0.01$ and the total time steps $10^4$. In addition, the lower bound and upper bound of branch resistance and reactance are $0.5$ times and $4$ times larger than the magnitude of the original value, respectively. Per unit systems are adopted with the base vale of power $100$ MVA in the simulation. The initial contingency is added on Branch $5$ connecting Bus $2$ and Bus $6$ (i.e., the red lightning in Fig.~\ref{bus24}), which enables the branch reactance to increase to $0.6$ and leads to the malfunction of TCSC on Branch $5$.
For simplicity, it is assumed that the injected bus power is subject to the disturbances, which satisfy the normal distribution with mean value of $0$ and standard deviation of $1$. The desired transmission power $\sigma$ is specified as the transmission power in the normal condition before the initial contingency.  Moreover, the function ``runpf" in Matpower is employed to solve the AC power flow equation and obtain the transmission power on each branch \cite{zim11}.

\begin{figure}
\scalebox{0.6}[0.6]{\includegraphics{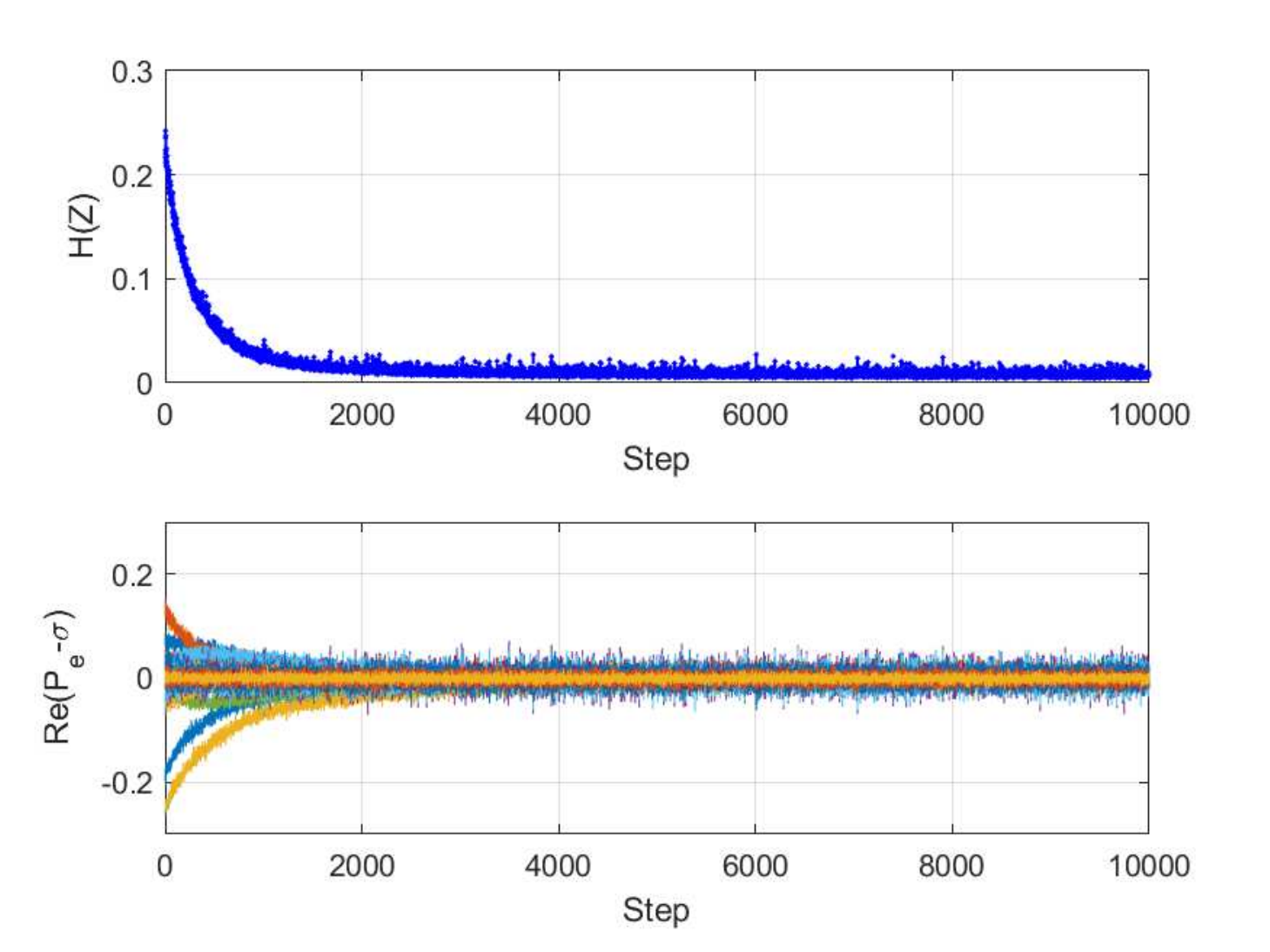}}\centering
\caption{\label{hpe} Evolution of objective function (upper panel) and mismatch of active power on each branch (lower panel).}
\end{figure}

\begin{figure}
\scalebox{0.6}[0.6]{\includegraphics{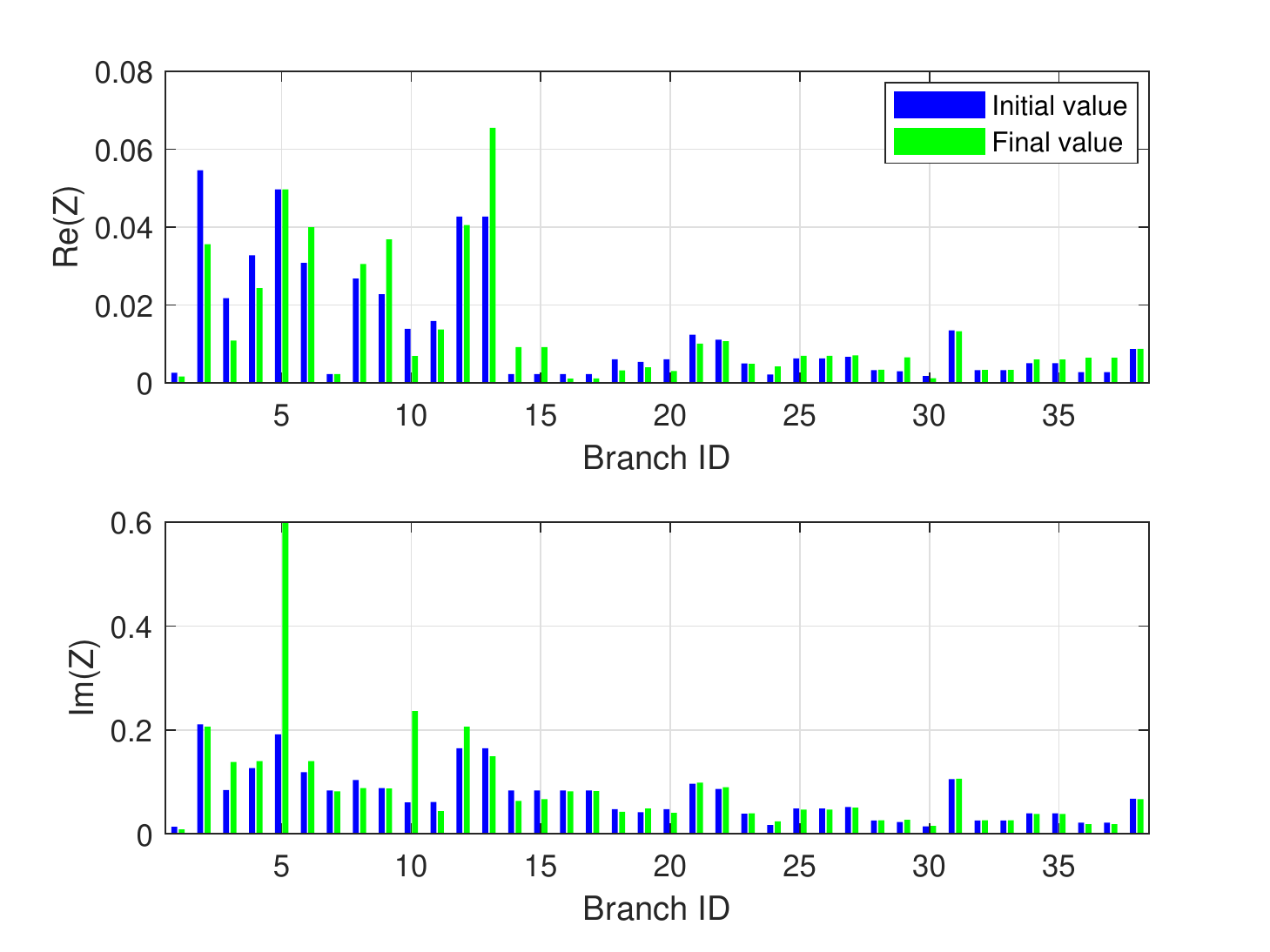}}\centering
\caption{\label{imp} Distribution of branch resistance (upper panel) and branch reactance (lower panel) at the initial step and the final step of numerical simulations.}
\end{figure}

\begin{figure}
\scalebox{0.6}[0.6]{\includegraphics{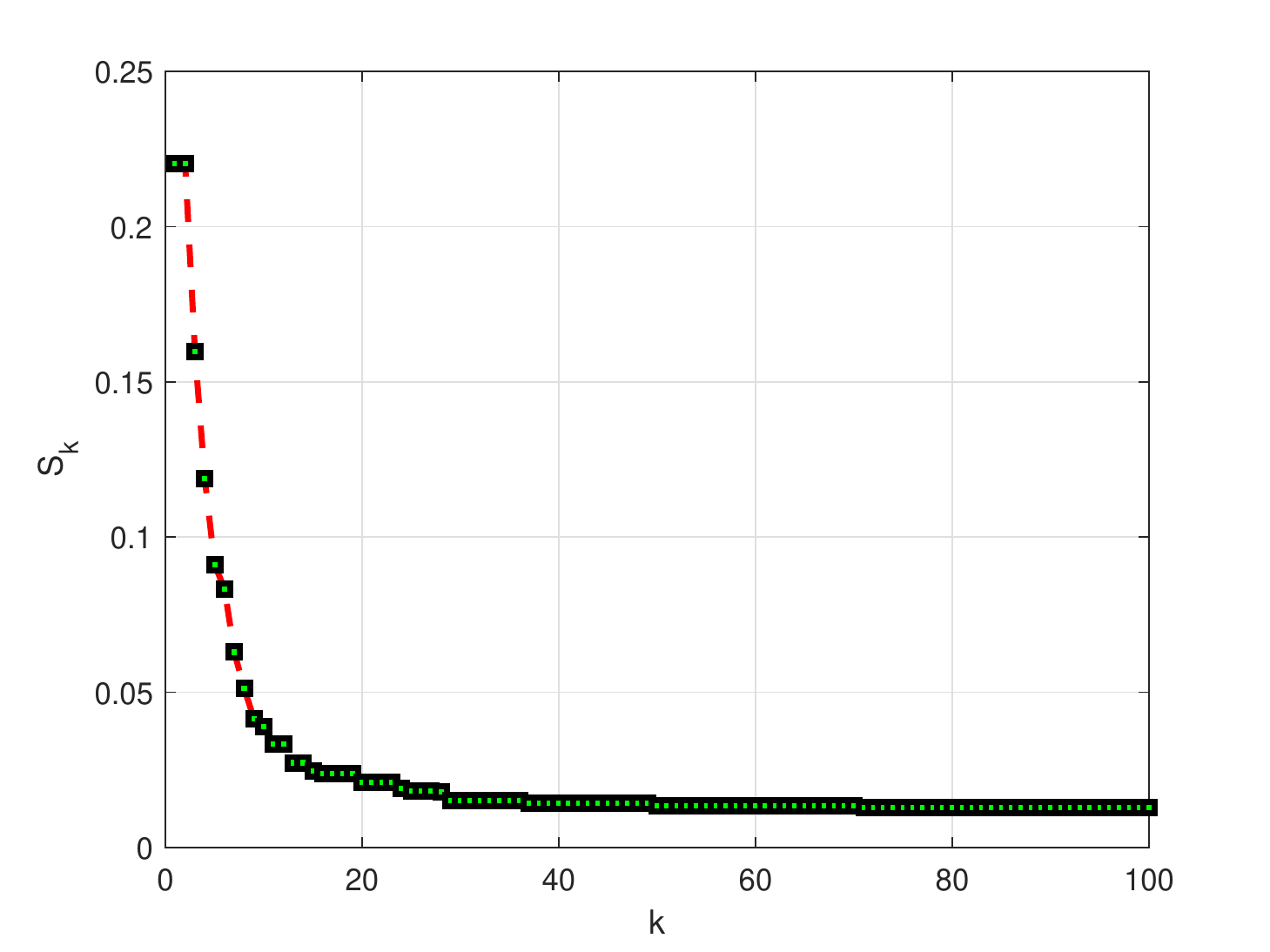}}\centering
\caption{\label{sk} Evolution of the performance index $S_k$.}
\end{figure}

The upper panel in Figure \ref{hpe} demonstrates that the objective function $H(Z)$ monotonously decreases from the initial value of $0.22$ to the final value of $0.006$ after $10^4$ time steps. The lower panel shows the trajectories of mismatch between the actual transmission power and the desired transmission power on each branch. It is observed that all the above trajectories gradually converge to zero with fluctuations as the evolution step increases. Thus, the two panels in Figure \ref{hpe} demonstrate the effectiveness of the cooperative control algorithm of TCSC to relieve the system stresses. Figure \ref{imp} presents the comparison of branch impedance (branch resistance in the upper panel and branch reactance in the lower panel) at the initial step and the final step in the numerical simulation using the cooperative control algorithm. Due to the malfunction of TCSC on Branch $5$, its branch resistance remains unchanged in the simulation. Specifically, the resistance values increase remarkably on Branch 6, Branch 8, Branch 9, Branch 13, Branch 14 and Branch 15, as opposed to the visible decrease of resistance value on Branch 2, Branch 3, Branch 4 and Branch 10.
In contrast, the resistance values fluctuate slightly on all other branches. In terms of branch reactance, there are no significant changes for all branches except for the remarkable increase on Branch $3$, Branch $5$ (due to the initial contingency), Branch $10$ and Branch $12$. Finally, Figure \ref{sk} shows the monotonous decrease of $S_k$ from the initial value of $0.22$ to the final value of $0.013$, which partially confirms the conclusion of Proposition \ref{pro_cca}. The total number of $k$ is $100$ since $S_k$ is computed for every $T=100$ steps with the total steps of $10^4$. In addition, the Jacobian matrix $J(Z)$ is updated for $79$ times in the simulation.

\section{CONCLUSIONS}\label{sec:con}
In this paper, we proposed a cooperative control algorithm for TCSC to relieve the stress of cyber-physical power systems. By slightly disturbing the branch impedance of power system, the control algorithm is able to estimate the elements of Jacobian matrix and cooperatively regulate the TCSC on each branch for the effective relief of power system stress. The proposed approach was validated by simulation results on IEEE $24$ Bus Systems with time-varying loads. Future work may include the optimal deployment of limited TCSC agents on branches and the estimation of Jacobian elements by analyzing the real PMU data without disturbing the branch impedance. In addition, we also plan to develop the distributed control algorithm for FACTS devices to enhance the resilience of cyber-physical power systems.

\section*{ACKNOWLEDGMENTS}
This work is partially supported by the Future Resilient Systems Project at the Singapore-ETH Centre (SEC), which is funded by the National Research Foundation of Singapore (NRF) under its Campus for Research Excellence and Technological Enterprise (CREATE) program. It is also partially supported by Ministry of Education of Singapore under contract MOE2016-T2-1-119.

\end{document}